\documentclass[conference]{IEEEtran}
\usepackage{cite}
\usepackage{url}
\usepackage{soul}
\usepackage{algorithm}
\usepackage[noend]{algpseudocode}

\usepackage{comment}

\PassOptionsToPackage{pdftex}{graphicx}
\usepackage{adjustbox}
\usepackage{amsmath}
\usepackage{amsfonts}
\usepackage{amssymb}
\usepackage{amsthm}
\usepackage{mathtools}
\usepackage{longtable}
\usepackage[utf8]{inputenc}
\usepackage{array}
\usepackage{comment}
\usepackage{subcaption}
\usepackage{pifont}

\usepackage{hyperref}
\usepackage{tabularx}
\usepackage{supertabular}

\usepackage{tikz}
\usetikzlibrary{patterns}
\usetikzlibrary{shapes}
\usepackage{xspace}

\usepackage{listings}
\usepackage{xcolor}
\usepackage{balance}
\definecolor{hblue}{HTML}{206899}		
\definecolor{cgreen}{rgb}{0,0.6,0}		
\definecolor{cpurple}{rgb}{0.58,0,0.82}	
\definecolor{cgray}{rgb}{0.5,0.5,0.5}	
\definecolor{cbg}{rgb}{0.975,0.975,0.975}	
\usepackage{color, colortbl}

\hypersetup{
	colorlinks     = true,
	linkcolor      = hblue,
    pdfcreator     = {\LaTeX{}},
    bookmarksopen  = true,
    bookmarksdepth = 3,
    pdfauthor      = {M. Pontecorvi, C. Segarra, M. Signorini, M. Caprolu, R. Di Pietro},
    pdftitle       = {CENTRIC: ClustRting bitcoiN TRansaction Chains in the blockchain},
    pdfsubject     = {Bitcoin Blockchain Parsing},
    pdfkeywords    = {}
}			
\DeclareCaptionFormat{listing}{%
    \parbox{.85\linewidth}{#1#2#3}%
}

\colorlet{punct}{red!60!black}
\definecolor{background}{HTML}{EEEEEE}
\definecolor{delim}{RGB}{20,105,176}
\colorlet{numb}{magenta!60!black}

\lstdefinelanguage{json}{
    basicstyle=\normalfont\ttfamily,
    numbers=left,
    numberstyle=\scriptsize,
    stepnumber=1,
    numbersep=8pt,
    showstringspaces=false,
    breaklines=true,
    backgroundcolor=\color{background},
    literate=
      {:}{{{\color{punct}{:}}}}{1}
      {,}{{{\color{punct}{,}}}}{1}
      {\{}{{{\color{delim}{\{}}}}{1}
      {\}}{{{\color{delim}{\}}}}}{1}
      {[}{{{\color{delim}{[}}}}{1}
      {]}{{{\color{delim}{]}}}}{1},
}

\usepackage{etoolbox}
\makeatletter
\patchcmd{\@makecaption}
  {\scshape}
  {}
  {}
  {}
\makeatletter
\patchcmd{\@makecaption}
  {\\}
  {.\ }
  {}
  {}
\makeatother
\def\tablename{Table}

\newcommand{\blockHeight}{591,872}

\newcommand{\graphName}{TIO}
\newcommand{\unknown}{unknown}

\newcommand{\cmark}{\ding{51}}%
\newcommand{\xmark}{\ding{55}}%

\usepackage{fp}

\usepackage{multirow}

\newtheoremstyle{break}
  {\topsep}{\topsep}%
  {\itshape}{}%
  {\bfseries}{}%
  {\newline}{}%
\theoremstyle{plain}
\newtheorem{lemma}{Lemma}
\newtheorem{remark}{Remark}

\theoremstyle{break}
\newtheorem{definition}{Definition}
\newtheorem{theorem}{Theorem}

\newcommand{\ndssColorTex}{\textcolor{black}}
\newcommand{\arXivColorTex}{\textcolor{black}}

\ifCLASSINFOpdf
\else
\fi
\begin{document}
%
\title{A Novel Framework for the Analysis of Unknown Transactions in Bitcoin: Theory, Model, and Experimental Results } 


\author{\IEEEauthorblockN{Maurantonio Caprolu\IEEEauthorrefmark{1},
Matteo Pontecorvi\IEEEauthorrefmark{2},
Matteo Signorini\IEEEauthorrefmark{2}, 
Carlos Segarra\IEEEauthorrefmark{3} and
Roberto Di Pietro\IEEEauthorrefmark{1}}\\
\IEEEauthorblockA{\IEEEauthorrefmark{1}Division of Information and Computing Technology, College of Science and Engineering\\Hamad Bin Khalifa University, Qatar Foundation - Doha, Qatar}
\IEEEauthorblockA{\IEEEauthorrefmark{2}NOKIA Bell Labs - 91620 Nozay, France}
\IEEEauthorblockA{\IEEEauthorrefmark{3}Imperial College, London, UK}
}


\maketitle
\sloppy

\begin{abstract}
Bitcoin (BTC) is probably the most transparent payment network in the world, thanks to the full history of transactions available to the public. Though, Bitcoin is not a fully anonymous environment, rather a pseudonymous one, accounting for a number of attempts to beat its pseudonimity using clustering techniques. There is, however, a recurring assumption in all the cited deanonymization techniques: that each transaction output has an address attached to it.  That assumption is false. An evidence is that, as of block height \blockHeight{}, there are several millions transactions with at least one output for which the Bitcoin Core client cannot infer an address. \\ 
In this paper, we present a novel approach based on sound graph theory for identifying transaction inputs and outputs. Our solution implements two simple yet innovative features:  it does not rely on BTC addresses and explores all the transactions stored in the blockchain. All the other existing solutions fail with respect to one or both of the cited features. 
\arXivColorTex{In detail, we first introduce the concept of Unknown Transaction and provide a new framework to parse the Bitcoin blockchain by taking them into account.}
Then, we introduce a theoretical model to detect, study, and classify---for the first time in the literature---unknown transaction patterns in the user network. Further, in an extensive experimental campaign, we apply our model to the Bitcoin network to uncover hidden transaction patterns within the Bitcoin user network. Results are striking: we discovered more than $30,000$ unknown transaction DAGs, with a few of them exhibiting a complex yet ordered topology and potentially connected to automated payment services. To the best of our knowledge, the proposed framework is the only one that enables a complete study of the unknown transaction patterns, hence enabling further research in the fields---for which we provide some directions.
\end{abstract}


%

\section{Introduction}\label{sec:introduction}

Known as the first successful virtual currency with the potential to disrupt the banking system and provide peer-to-peer payments, Bitcoin has been widely adopted in ransomware campaigns such as \textit{Wannacry}~\cite{wanacry-wiki} and \textit{NotPetya}~\cite{notpetya-wiki}. The main reasons being the relative diffusion of Bitcoin, as well as  privacy. Privacy 
is easily achievable with Bitcoin pseudonyms in the form of randomly generated addresses that can be used to send/receive money without being linked to any real identity. 

Being organized into wallets, Bitcoin’s addresses can be easily and freely self-generated by end-users (neither banks nor trusted third parties are needed)  without any limitation on their number. Indeed, using each address for a single transaction is a strongly advised common practice in the Bitcoin community and has always been its key component in providing privacy, since the resulting transactions’ network has always been assumed too complex to track money flows. 

Such ``privacy through complexity'' approach has been enhanced in the last years by online services such as mixers and tumblers.
Similar to the TOR project~\cite{Jawaheri2018}, aimed at concealing  user's location and network activities from anyone conducting network surveillance or traffic analysis, mixers and tumblers try to conceal users’ addresses that took part in some monetary transactions.
The functioning of such services is quite simple and, similarly to TOR, they require to bounce bitcoins through peers in order to make their tracking hard. Furthermore, the bitcoins that are bounced and returned to the user are not the same that were initially sent, since they come from other sources (i.e. other addresses). 
However, during the last few years, it has been shown that the above ``privacy through complexity'' approach can be attacked by clustering the addresses into groups that are likely to belong to the same entity (a user, a shop, a mixer etc.). In 2015, David Nick described~\cite{Nick2015} some of the most famous heuristics being used still to date for Bitcoin address clustering, such as the shadow, consumer, optimal-change, and multi-input. These heuristics, applied to transactions data previously extracted from the blockchain using parsing algorithms, produce in output the clustered user network.

All such parsing algorithms suffer from the cognitive bias that Bitcoin transactions are just operations that link one address to another~\cite{Maesa2016,Androulaki2013}. However, such address-based linkability is not enforced directly in the Bitcoin protocol: the protocol only verifies that the locking and unlocking scripts do not produce any false statement~\cite{Antonopoulos2017}. 
This cited bias is also reinforced by the fact that most of the locking and unlocking scripts follow a specific pattern based on asymmetric encryption using randomly generated addresses. However, such a pattern is not mandatory. 
Consequently, transactions data parsed relying only on recognizable locking/unlocking scripts risk to be incomplete, and they often are. This incompleteness, in turn, causes a loss of reliability in all modern clustering and de-anonymization techniques. Moreover, intentionally crafted custom transactions could be used to hide illicit money flows, while being completely invisible to modern automatic parsers unable to decode output addresses.
As an evidence, note that up to block with height \blockHeight{}, Bitcoin's blockchain contains more than $22$ million transactions with at least one locking script (or output) not following any well-known locking/unlocking script. We will often refer to these outputs as \emph{\unknown{}} transaction outputs.

\subsection{Contribution}
Our contributions are both theoretical and applied. At glance, we are the first, to the best of our knowledge, to provide a navigation tool within the Bitcoin blockchain that: does not rely on addresses and explores all the transactions stored in the blockchain. In detail, our contributions  can be summarized as follow: 

\begin{itemize}
    \item \arXivColorTex{we introduce the general concept of \emph{\unknown{}} transactions, subsuming the less rigorous definition of non-standard transactions provided by the Bitcoin protocol;}
    \item we extend the definition of user network $\mathcal{U}$ by including \emph{\unknown{}} transactions;
    \item we design a novel (theoretically sound) parsing methodology to study, and likely understand for the first time, \emph{\unknown{}} transaction patterns by using a specific class of graphs called T-DAGs; 
    \item we show that T-DAGs can be efficiently compared via isomorphism, offering a new mechanism for clustering similar transaction patterns.
    \item we test our algorithms over the Bitcoin network, \arXivColorTex{collecting and analysing all \unknown{} transactions in the ledger from Bitcoin origins until block \blockHeight{};} 
    \item \arXivColorTex{we further refine our results by removing trivial non-standard transactions already observed in the literature, revealing } classes of hidden \emph{\unknown{}} transaction patterns never considered before;
    \item to the best of our knowledge, this is the first study of the Bitcoin transaction network which includes \emph{\unknown{}} transactions. 
\end{itemize}

\arXivColorTex{With reference to the Bitcoin context, our contribution can be used to collect and observe the existing \emph{\unknown{}} transaction patterns in the ledger, as well as those that will be generated in the future. These patterns, previously ignored by any transactions analysis work, can be used to complete the user network $\mathcal{U}$ and the transaction network $\mathcal{T}$. Our methodology can be applied to any clustering and de-anonymization technique to improve its effectiveness by leveraging a complete and reliable Bitcoin transactions database.}
Moreover, our novel solution can be extended immediately to any other system where transaction patterns can be modeled using T-DAGs.\\ 

\textbf{Organization:} The rest of the paper is organized as follows. In Section~\ref{sec:background} we introduce preliminary concepts used in the paper. In Section~\ref{sec:related-work} we summarize the state of the art for Bitcoin addresses clustering and graph isomorphism problem. In Section~\ref{sec:gen-wflow} we define our theoretical model capable of extending user network in order to identify those sub-graphs that involve \unknown{} transactions. In Section~\ref{sec:implementation}, we describe the application of our theoretical model to the Bitcoin network, discussing the results achieved during our experimental evaluation. Finally, in Section~\ref{sec:conclusion} we summarize our results, discuss the possible implications, and suggest future research lines.

\section{Background} \label{sec:background}

\subsection{Preliminaries}
Let $G = (V, E)$ be a graph with nodes $V$ and edges $E$.
We will consider both undirected graphs with $E \subseteq \binom{V}{2}$, and directed graphs with $E \subseteq V \times V$. 
An undirected graph $G$ is said to be connected if, for any $u, v \in V$, there is a path from $u$ to $v$. 
In a directed graph, edges are sometimes called arcs and the first vertex of an arc is referred to as the tail and the second one as the head.
Given a directed graph $D = (V, E)$, the underlying graph $G = (V, E')$ is the undirected graph obtained with the same vertex set and an edge set corresponding of each arc in $E$. That is, if $(u, v)$ is an arc in $D$ then both $(u,v)$ and $(v,u)$ are edges in $G$.
A directed graph is weakly connected if the underlying graph is connected.
From now on, we will always refer to directed graphs and use edges and arcs interchangeably.

For an edge $e = \{u,v\} \in E$, we say that $u$ is adjacent to $v$, and vice-versa. If there exists a sequence of edges connecting two vertices $u$ and $v$, we say that these two vertices are reachable. If $u$ is reachable starting from  itself (through a non-empty sequence of edges) 
we say $G$ contains a cycle. If $G$ contains no cycles, $G$ is said to be acyclic. We call directed acyclic graphs DAGs. A tree is an undirected, connected, and acyclic graph. We use $K_{x,y}$ to indicate the complete bipartite graph between two sets of nodes $X$ and $Y$ with $|X|=x$ and $|Y|=y$.
Given a vertex $v \in V$, $v$'s indegree, $\Gamma^-(v)$, is the cardinal of edges (or arcs) that have $v$ as its head. Symmetrically, $v$'s outdegree, $\Gamma^+(v)$, is the cardinal of edges (or arcs) that have $v$ as its tail. We call a vertex with zero indegree a source, and one with zero outdegree a sink.
Given two graphs $G_1 = (V_1,E_1)$ and $G_2=(V_2,E_2)$, an isomorphism is a bijection function $f(.)$ between the sets of vertices, $V_1$ and $V_2$, such that it maintains adjacency. That is, $(u_1, u_2) \in E_1 \Rightarrow (f(u_1), f(u_2)) \in E_2$. A graph labeling is the assignment of labels to the vertices, edges, or both, of a graph. Given a graph $G$, a canonical form is a labeled graph $\bar{G}$, isomorphic to $G$, such that $H$ is isomorphic to $G$ iff $\bar{G} = \bar{H}$---i.e, they have identical canonical forms. 
An equivalence relation $\sim$, is a binary relation on a set $S$  that is: (i) reflexive; (ii) symmetric; and, (iii) transitive. For any elements $a, b$ and $c$ in $S$ it holds: (i) $a \sim a$, (ii) $a \sim b \Rightarrow b \sim a$ and (iii) $a \sim b \wedge b \sim c \Rightarrow a \sim c$. 
An equivalence relation, $\sim$, partitions a set in equivalence classes. 
A total ordering, $\prec$,  is a binary relation on a set $S$ that is: (i) antisymmetric; (ii) transitive; and, (iii)  connex. That is, $\forall a, b, c \in S$: (i) $a \prec b \wedge b \prec a \Rightarrow a = b$; (ii) $a \prec b \wedge b \prec c \Rightarrow a \prec c$; and, (iii) $a \prec b \cup b \prec a$.

\subsection{Bitcoin Transaction Network and User Network}
The Bitcoin network is composed by peers which can be either \emph{full nodes} or \emph{clients}, with full nodes validating transactions and blocks~\cite{Bitcoin_full_node}.
The Bitcoin protocol stores every transaction in a publicly distributed ledger, thus allowing everyone to read the transactions occurred among Bitcoin users. 
To ease the analysis of the money flow contained in the Bitcoin ledger, the Bitcoin data-set can be modeled with two different graphs (Transaction and User Network) as introduced  in~\cite{Fergal2013}.\\*
The \emph{Transaction Network} $\mathcal{T}$ describes the Bitcoin flow between transactions over time \cite{Motamed2019}. Each node represents a transaction, and each directed edge between two nodes n1 and n2 represents a money flow that is an output for n1 and an input for n2. The network is a directed acyclic graph (DAG) since the output of a transaction can never be an input (either directly or indirectly) to the same transaction.\\ 
The \emph{User Network} $\mathcal{U}$ describes the Bitcoin flow between users over  time \cite{DiFrancescoMaesa2016}. Each node represents a Bitcoin user, and each directed edge between two nodes $n_1$ and $n_2$ represents an inputs-outputs pair of a single transaction, where the inputs belongs to $n_1$ and the outputs belongs to $n_2$. More in detail, the \emph{User Network} is a weighted directed hypergraph $\mathcal{U} = (\mathcal{A}, \mathcal{T})$ where $\mathcal{A}$ is the set of all the addresses used in the Bitcoin network, and $\mathcal{T}$ is the set of transactions. Every transaction $t \in \mathcal{T}$  can be modeled as a pair of ordered sets $(X,Y)$ with $X,Y \subseteq \mathcal{A}$ , where addresses included in $X$ are inputs of $t$ and addresses in $Y$ are outputs of $t$.

\begin{table*}[]
\centering
\begin{tabular}{|c|c|c|c|c|c|c|}
\hline
\multirow{2}{*}{\textbf{Solution}} & \multicolumn{2}{c|}{\textbf{Objective}} & \multicolumn{2}{c|}{\textbf{Unknown Transactions}} & \multicolumn{2}{c|}{\textbf{Considered Ledger Portion}} \\ \cline{2-7} 
 & \textbf{Parser} & \textbf{Clustering} & \textbf{Considered} & \textbf{Study of Patters} & \textbf{From Block} & \textbf{To Block} \\ \hline
\textbf{Blocksci~\cite{Kalodner2017}} & \cmark & \cmark & \xmark & \xmark & 0 & custom \\ \hline
\textbf{Blockchain.com} & \cmark & \cmark & \xmark & \xmark & 0 & last mined block \\ \hline
\textbf{Bitcoincore} & \cmark & \xmark & \xmark & \xmark & 0 & last mined block \\ \hline
\textbf{~\cite{Fergal2013}} & \xmark & \cmark & \xmark & \xmark & 0 & 130367 \\ \hline
\textbf{~\cite{Androulaki2013}} & \cmark & \cmark & \xmark & \xmark & 0 & 140000 \\ \hline
\textbf{~\cite{Meiklejohn2013}} & \cmark & \cmark & \xmark & \xmark & 0 & 231207 \\ \hline
\textbf{~\cite{kondor2014rich}} & \xmark & \xmark & \xmark & \xmark & 0 & 235000 \\ \hline 
\textbf{\cite{spagnuolo2014bitiodine}} & \cmark & \cmark & \xmark & \xmark & 0 & 267350 \\ \hline
\textbf{\cite{Maesa2016}} & \xmark & \cmark & \xmark & \xmark & 0 & 389799 \\ \hline
\textbf{\cite{Bartoletti2017}} & \xmark & \xmark & \cmark & \xmark & 0 & 453200 \\ \hline 
\textbf{\cite{Ermilov2017}} & \xmark & \cmark & \xmark & \xmark & 0 & 456520 \\ \hline
\textbf{Our Solution} & \cmark & \cmark & \cmark & \cmark & 0 & 591872 \\ \hline
\end{tabular}
\caption{\ndssColorTex{Comparison of our solution with popular parsers and state-of-the-art clustering techniques.}} 
\label{tab:comparison}
\end{table*} 

\section{Related Work} \label{sec:related-work}
\arXivColorTex{To the best of our knowledge, very few works have investigated \unknown{} transactions in the Bitcoin ledger. In~\cite{Bartoletti2017},} the authors analyzed 1,887,708 transactions containing the \texttt{OP\_RETURN} instruction. They found that $15\%$ of them are empty transactions, generated by different activities on the Bitcoin network, such as stress tests or DoS attacks. The remaining transactions are not empty but, similarly to the previous ones, they are not used for transferring funds. In fact, they have a different, specific goal: to store data in the Bitcoin ledger. The \texttt{OP\_RETURN} transactions do not have a valid recipient, since they are not used to transfer funds. Therefore, they cannot be redeemed. For this reason, these transactions are not of particular interest for clustering and de-anonymizing techniques of Bitcoin users, i.e., they are present neither in $\mathcal{T}$ nor in $\mathcal{U}$.
\arXivColorTex{A more in-depth analysis of non-standard transactions in the Bitcoin network has been proposed in~\cite{bistarelli_2019}. The authors explored the ledger collecting and classifying both standard and non-standard transactions to understand why users sometimes do not adhere to the protocol. To achieve this goal, they mainly focus on analyzing non-standard transactions, classifying them into nine different typologies.\\Although these studies analyzed some non-standard transactions, their purpose is only to analyze the semantic, considering every transaction as a stand-alone object. Consequently, such transactions are still ignored in the construction of both the user network $\mathcal{U}$ and the transaction network $\mathcal{T}$, leading to incomplete and possibly unreliable data structures. To solve this problem, we first collected all the unknown transactions in the Bitcoin ledger, regardless of their semantic.}
We then focused on those that have a valid locking script as they have an impact on the de-anonymization and clustering techniques, neglected by all previous works in the field (see a summary in Section~\ref{subsec:related-deanonymization}). By using the proposed methodology, our framework is able to correctly parse \unknown{} transactions, identify their patterns, and complete the user network $\mathcal{U}$ and the transaction network $\mathcal{T}$ with additional data never considered before.


\subsection{Bitcoin Addresses Clustering} \label{subsec:related-deanonymization}
Recently, security properties of blockchain-based protocols, with particular attention on the Bitcoin network, received increasing attention~\cite{Gervais:2016:SPP:2976749.2978341, Kiffer:2018:BMA:3243734.3243814, Kwon:2017:SAD:3133956.3134019, Gervais:2015:TDB:2810103.2813655, Dagher:2015:PPP:2810103.2813674}.
Among the different topics covered, several methods have been proposed to cluster bitcoin addresses~\cite{Harrigan2016}. These contributions build $\mathcal{U}$ by applying different heuristics to $\mathcal{T}$. The most used heuristic, the \texttt{Common Input}, is based on the observation that all the inputs of a multi-input transaction belong to the same user~\cite{Fergal2013}. This heuristic has been used to cluster Bitcoin addresses with the aim of deanonymizing users by using off-chain data~\cite{Lischke2016} and study the Bitcoin network to uncover different properties~\cite{Maesa2016, Ober2013, kondor2014rich}. Later, a more advanced heuristic called \texttt{One-time Change} was introduced~\cite{Meiklejohn2013, Ermilov2017, Androulaki2013, spagnuolo2014bitiodine}, based on the detection of the change among the output addresses.
\par All of these approaches built the transaction network $\mathcal{T}$ and the user network $\mathcal{U}$ by parsing the Bitcoin blockchain using the reference implementation (i.e. Bitcoin Core) or other blockchain explorer software. As explained in Section~\ref{subsec:gwflow-1}, current parsing algorithms, due to some incomplete assumptions, do not consider \unknown{} transactions, providing incomplete and in some cases incorrect data to clustering and deanonymization algorithms.\\
Table~\ref{tab:comparison} shows a comparison between the p0roposed solution and: (i) the Bitcoin reference implementation, i.e., Bitcoincore; (ii) blockchain explorer services, represented by the most popular one, blockchain.com; (iii) other platforms to analyze the Bitcoin blockchain, e.g., Blocksci; and, (iv) the most representative contributions in the literature that already used the user network $\mathcal{U}$ or the transaction network $\mathcal{T}$.

\subsection{Graph Isomorphism} \label{subsec:related-graph}
The graph isomorphism (GI) problem consists in determining whether given two graphs $\mathcal{G}$ and $\mathcal{V}$ there exists a bijection between both sets of vertices that preserves adjacencies. 
The lowest time bound for general GI stood since 1983 and until very recently at $\exp{\mathcal{O}\left(\sqrt{n \log n}\right)}$~\cite{Babai1983}.
However, polynomial time algorithms have long been known for different families of graphs. Among others, for graphs with bounded valence, the first polynomial time test dates from 1980~\cite{Luks1980}. This work was used to solve the problem for graphs with bounded eigenvalue multiplicity~\cite{Babai1982} in $\mathcal{O}\left(n^{4m+c}\right)$. For circulant graphs, a $\mathcal{O}(n^2)$ testing algorithm is known~\cite{Muzychuk2004}. The tree isomorphism problem was proven~\cite{Hopcroft1974} to be linear in the number of vertices, that is, two trees $T_1$ and $T_2$ can be tested for isomorphism in $\mathcal{O}(n)$ comparing them in a bottom-up fashion. A result from 2016~\cite{Babai2016} improved the 1983 bound for the GI problem to quasipolynomial time. The author proved an upper bound of $\exp{\left(\log n \right)^{\mathcal{O}(1)}}$ for testing wether two arbitrary graphs are isomoprphic or not.

A possible test of graph isomorphism is performed through canonical forms and canonical labelings. A canonical form of a graph is a representative of a class of graphs closed under isomorphisms, and with a linearly ordered vertex set~\cite{BabaiLuks1983}. That is, two graphs are isomorphic if and only if they yield the same canonical form. A \emph{canonical labeling} is a string derived from a canonical form of a graph obtained through a defined mapping. The problem of obtaining canonical labelings for general graphs remains quasi-exponential in the number of vertices. A combinatorial method that runs in $\exp{n^{2/3}+o(1)}$ has long been known~\cite{BabaiLuks1983}. In the same paper the authors give a procedure to compute the canonical form of graphs of bounded valence (or degree) in polynomial time. Due to its relevant applications and the urge for practical solutions, implementations of graph canonization packaged in different software has proven to be successful. The McKay canonical labeling algorithm~\cite{McKay1978} packaged in the \texttt{nauty} software is one of the most celebrated solutions. It was the most popular solution from the early 80s until 2004, when a modern re-implementation~\cite{Darga2004}, \texttt{saucy}, was introduced. Exploiting symmetry and sparsity in the search space pruning the latter managed to outperform the former by several orders of magnitude. The refinements later introduced in \texttt{bliss}~\cite{Junttila2007} concluded in the release of \texttt{Traces}. The definition of the latter piece of software together with a comparison of all the software packages presented is discussed by its authors in~\cite{McKay2013}.

When restricting the problem to rooted or un-rooted trees, linear is the best possible sequential runtime~\cite{Hopcroft1974}. However, tree canonization has been proven to be in alternating logarithmic time~\cite{Buss1997}.
The labeling therein defined is depth-first and the approach presented in this paper is breadth-first. In spite of that, the ordering relation we define is based on the same principles. However, a more recent approach~\cite{Yun2015} for tree canonical labeling uses a breadth-first approach. The authors therein do provide a less verbose labeling (at most $3n$) than ours, but this is due to the fact that they deal with specialized trees: ones where each vertex has an in-degree of at most 1. 

\section{Unknown Transactions Recognition} \label{sec:gen-wflow}
In this section, we introduce the general concept of \emph{\unknown{}} transactions, their classification, and the theoretical model to identify patterns within \emph{\unknown{}} transactions.

\subsection{Unknown Transactions and Working Framework} \label{subsec:gwflow-1}
\arXivColorTex{The Bitcoin protocol provides its community with standard templates that must be used to create the locking and unlocking scripts that make up a transaction. The use of such templates is then enforced by miners using two functions, isStandardTx() and isStandard(), which check the compliance of each transaction's inputs and outputs, respectively. In fact, a transaction is considered standard, and therefore accepted by the network, only if both functions return TRUE. If even one of them returns FALSE, the transaction is considered non-standard and discarded. This mechanism should prevent any use of Bitcoin transactions other than the ones conceived by the protocol, to avoid the spread of malicious transactions. However, even non-standard transactions can be included in the blockchain, thanks to miners who relax these controls~\cite{bistarelli_2019}.\\ Similar to the concept of non-standard transaction, we define \emph{\unknown{}} transactions as follow:}


\begin{definition}[Unknown Transaction]\label{def-unkn_tx}
We call a transaction (TX) \emph{\unknown{}} if it contains an input or an output with a \texttt{Null} value address, i.e. not correctly identified by the Bitcoin Core client.
\end{definition}

\noindent
\arXivColorTex{This definition embraces a set of Bitcoin transactions, of which non-standards are currently a subset, regardless of what the protocol considers standard or non-standard. The concept of unknown transactions allows us to protect our framework from future variations of the Bitcoin protocol and guarantees compatibility with other systems.\\}
Among \unknown{} TXs, we are only interested in transactions whose \emph{output} contains a \texttt{Null} address. This is because the address of an input is unequivocally determined by that of the output that it is spending. Hence, resolving the address of the output is equivalent to resolving the address of the associated input as well. Additionally, an input only exists to fund the outputs contained in its corresponding transaction. Thus, there is no such thing as an address for an input as it does not belong to someone. However, for simplicity, we agree to assign to an input the same address held by the output that it is spending.
To unambiguously process the list of transaction hashes, we need to uniquely identify inputs and outputs. This is done using the \emph{Unique Transaction Input-Output Identifier}, introduced below in Definition \ref{def:utxio}. But, before introducing such an Fidentifier, we motivate why it is required by our approach.

For elaborated blockchain analysis, it is a good idea to initially parse all the data from a running miner and, once the data is organized in a more accessible manner, apply further and more complex post-processing. However, we have discovered several imprecisions in the blockchain's parsing process:

\begin{itemize}
    \item[(i)] \textbf{Excess of abstraction} Blockchain parsers such as BlockSCI~\cite{Kalodner2017} introduce a completely new level of abstraction over the one already specified in the reference implementation~\cite{Bitcoin_develop}. 
    Defining new wrappers, lots of different classes, and incomplete references can make a parser difficult to use and debug. Additionally, we discovered that, in the particular case of BlockSCI, there are errors in their parsing methodology; 
    \item[(ii)] \textbf{Excess of identifiers} Bitcoin's blockchain is an environment based on uniqueness. Every item must be uniquely identified and hashing algorithms already provide a way to do so. However, some parsers~\cite{Kalodner2017} insist on giving an alternative enumeration for transactions and addresses. This makes databases harder to navigate and makes it non-intuitive to mimic the client's behavior or debug the processed data;

    \item[(iii)] \textbf{Using Public Keys as Identifiers} As introduced in Section~\ref{sec:introduction}, to the best of our knowledge, 
    existing works~\cite{Maesa2016,Androulaki2013}, clustering Bitcoin Addresses to find real end-users, assume that each transaction output must have an address assigned to it. 
    This is false. In fact, up to block with height $481823$, Bitcoin's blockchain contains $3255688$ \unknown{} transactions.
\end{itemize}

Given the above problems, our proposed framework aims to provide a reference for storing Bitcoin's data in a database; minimizing the amount of abstraction involved, reusing whenever possible the identifiers provided by the reference implementation, and keeping the structure simple and clear. 

To be consistent with both (i) and (ii), we only introduce the critical functionalities not covered by the Bitcoin core~\footnote{The core client cannot find which transaction input is spending a given unspent output.}.
This way, parsers using our framework can be easily compared against each other, the Bitcoin Core Client, or even web explorers. 
Further layers of abstraction depending on the application should be detached from the parsing phase to avoid situations where complex post-processing is discredited by incorrect data parsing. This means that all non-relevant information for blockchain navigation is not included in the framework.

\textbf{The Framework}
Our framework uses only minimal abstraction and provides a robust, reliable, and fast way to navigate through the Bitcoin's blockchain. It is also easily portable: all applications that query or do some sort of post-processing with Bitcoin's data can use it. To fulfill these conditions and to preserve minimality, only the necessary attributes are included.
All other features included in the reference implementation, that provide key information about each transaction, but do not improve the exploration of the blockchain, are discarded. In fact, they can be easily obtained by using the identifiers provided by our framework, together with any Bitcoin client.

We present a database layout that only contains two types of entities: \texttt{block} and \texttt{tx}.
\begin{itemize}
    \item[(i)] \textbf{\texttt{block}:} represents a block in the blockchain. It is uniquely identified by two parameters: \texttt{hash} and \texttt{height}. Both the parameters can be used to retrieve a \texttt{block} element from the database without ambiguity. Each \texttt{block} element has an additional parameter, \texttt{tx}. \texttt{tx} is an array of hashes, each one referencing a transaction
 included in the block (see next item for a description on the \texttt{tx} entity). Each element in the array can be uniquely identified, and accessed, by the index of their position within the array. This way, the $m$-th transaction of the $n$-th block can be identified without uncertainty;
    
    \item[(ii)] \textbf{\texttt{tx}:} represents a confirmed transaction (TX) in the blockchain. Since the Bitcoin's ledger contains different cases of transactions with the same hash\footnote{Blocks 91812 and 91842 contain a transaction with hash: ``d5d27987d2a3dfc724e359870c6644b40e497bdc0589a033220fe15429d88599''.}, this attribute cannot be used as a unique identifier. We realized that the Bitcoin Core client still uses the hash attribute to uniquely identify a transaction, causing a loss of information: searching for a particular transaction hash, the Core client returns only the last occurrence of that hash in the ledger. As a result, any transaction stored in the blockchain with a hash equal to a more recent transaction will never be returned by the client. For this reason, we uniquely identify a transaction using its attribute pair \texttt{<blockhash, hash>} which represents the hash of the block they belong to and its hash, respectively. The \texttt{vin} attribute is an array of pairs \texttt{<txID, txID[vout]>}. It represents the set of inputs contained in the transaction. Each input can be uniquely identified by their index within the transaction input array~\footnote{Since inputs do not have a global unique identifier, and since the reference client implementation does not define input entities, we have chosen not to do so either.}. The \texttt{txID} attribute from the pair is the \texttt{hash} attribute of the TX that contains the output that the input is spending and \texttt{txID[vout]} is the index of the spent output within the TX that contains it. Symmetrically, the \texttt{vout} attribute is an array of pairs \texttt{<txID, txID[vin]>} where, if the output is spent by some input in the future, the TX hash, and the index within the transaction where the output is spent, are included. If the TX is unspent, both attributes are set to \texttt{null}. Each output can be uniquely identified through the hash of the transaction they are contained in and their index in the output array (\texttt{tx.vout}).
\end{itemize}
The above introduced  structure leads to Definition~\ref{def:utxio}, which we will use often in the rest of the paper.
\begin{definition}[\graphName{}] \label{def:utxio}
    A \textbf{Unique Transaction Input-Output} (\graphName{}) is an identifier that can uniquely identify all the inputs and outputs contained in confirmed transactions within the blockchain. We denote the set of all inputs and outputs as $<\text{\graphName{}} >$.
\end{definition}
A first contribution of our framework is the possibility to \textit{travel to the future} in the blockchain. This allows us to easily identify the paths followed by bitcoins through the blockchain history. Definition~\ref{def:spending-funded} formalizes some new terminology related to our traveling mechanism.
\begin{definition}[Traveling the Blockchain] \label{def:spending-funded}
    Given a \graphName{}, we define the current terms:
    \begin{itemize}
        \item[(i)] If the \graphName{} corresponds to an \textbf{Input}:
            \subitem (a) The \textbf{spending output} of \graphName{} refers to the output that this input is using;
            \subitem (b) The \textbf{funded outputs} of \graphName{} refers to the outputs that this input is providing bitcoins to. By Bitcoin design, we assume that the funded outputs for an input are all the outputs contained in the same transaction than the input.
        \item[(ii)] If the \graphName{} corresponds to an \textbf{Output}:
            \subitem (a) The \textbf{spending inputs} of \graphName{} are all the inputs that funded this output. By Bitcoin design, we assume that all the spending inputs for an output are all the inputs contained in the same transaction than the output.
            \subitem (b) If the output is spent, the \textbf{funded input} is the input that is spending the output. Note that, this input will appear in a more recent transaction than the one containing the output.
    \end{itemize}
\end{definition}
A brief summary of the data structure used in our framework is provided in Table~\ref{table:standard}.
\begin{table}[!ht]
    \centering
    \caption{Summary table of our data structure.}
    \label{table:standard}
    \begin{adjustbox}{max width=\linewidth}
    \begin{tabular}{|l|l|l|} \cline{1-1} \cline{3-3}
        \cellcolor{gray!20}\textbf{\texttt{block}} & & \cellcolor{gray!20}\textbf{\texttt{tx}}\\[3pt] \cline{1-1} \cline{3-3}
        $\ast$ \texttt{hash} & & $\ast$ \texttt{hash}\\[3pt]
        $\ast$ \texttt{height} & & $+$ \texttt{blockhash} \\[3pt]
        $+$ \texttt{tx} $\coloneqq \left[ +\text{ \texttt{hash} } \right]_{<n> (\ast)}$ & & $+$ \texttt{vin} $\coloneqq \left[ \begin{array}{c} +\text{ \texttt{txID} } \\ +\text{ \texttt{txID[vout]} } \end{array} \right]_{<n> (\ast)}$ \\[3pt] \cline{1-1}
        \multicolumn{2}{c|}{} & $+$ \texttt{vout} $\coloneqq \left[ \begin{array}{c} +\text{ \texttt{txID} } \\ +\text{ \texttt{txID[vin]} $(\#)$ } \end{array} \right]_{<n> (\ast)}$ \\[3pt] \cline{3-3}
        \multicolumn{3}{c}{} \\[10pt] \hline
        \multicolumn{3}{|l|}{\textbf{Legend:}} \\[3pt]
        \multicolumn{3}{|l|}{\textbf{$\ast \coloneqq$} Attribute is a unique identifier for the entity.} \\[3pt]
        \multicolumn{3}{|l|}{\textbf{$+ \coloneqq$} Attribute of a given entity.} \\[3pt]
        \multicolumn{3}{|m{9.5cm}|}{\textbf{$\# \coloneqq$} New attributes that do not appear in the reference implementation.} \\[8pt]
        \multicolumn{3}{|m{9.5cm}|}{\textbf{$\left[ \cdots \right]_{<n> (\ast)} \coloneqq$} Array of elements with the attributes specified between brackets. These elements can be uniquely identified within their container by their position in the array (indexed by an integer $n$).} \\[3pt] \hline
    \end{tabular}
    \end{adjustbox}
\end{table}

\subsubsection*{Locking Script}

In addition to its \graphName{}, we are also interested in the locking script for an output. By \emph{locking script}, we refer to the script that has to be redeemed in order to spend the output. In the Bitcoin Core reference implementation, it is referred as \texttt{scriptPubKey}.\\
\arXivColorTex{In the later stages of our methodology, we will use TIOs to build the Unknown TX T-DAGs. Instead, the locking scripts will be used to filter our results by removing T-DAGs generated by transactions with purposes other than the transfer of crypto coins.}

\subsection{Unknown TX T-DAG Construction} \label{subsec:gwflow-3}

In this section, we lay the blockchain data in a graph using the framework defined in~\ref{subsec:gwflow-1}, introduce the concept of \textbf{Unknown TX graphs} and study the derived \textbf{T-DAGs}. The study of these directed graphs will enable us to describe, tailor and identify \unknown{} transaction patterns on the Blockchain.
\begin{definition}[\graphName{} graph] \label{def:utxio-graph}
    Let $<\text{TX}>$ be the set of confirmed transactions in the blockchain. Let $G_{\text{\graphName{}}} = (V,E)$ be a directed unweighted graph such that:
    \begin{itemize}
        \item[(i)] $V = \hspace{3pt} <\text{\graphName{}}>$
        \item[(ii)] $E = \left\lbrace \bigcup\limits_{t \in <TX>} \left\lbrace E \left( K_{|I_t|, |O_t|} \right) \hspace{-8pt} \bigcup\limits_{\substack{o \in O_t \\ o \not\in \text{UTXO}}} \hspace{-8pt} \left( o, \text{gFI}(o) \right) \right\rbrace \right\rbrace$
    \end{itemize}
    where \text{gFI} returns the funded input of a given output, given a transaction $t$, $I_t$ and $O_t$ denote $t$'s set of inputs and outputs respectively, and UTXO is the unspent transaction output store.
\end{definition}
\begin{lemma}\label{lemma:is-a-dag}
    The \graphName{} graph, $G_{\text{\graphName{}}}$, is a directed acyclic graph (DAG).
\end{lemma}
\begin{proof}
    Nodes in the graph represent validated inputs or outputs in the blockchain. This means that, when they were broadcast to the network, each miner checked them. For an input or an output to be validated, they must always point to an event that happened in the past. Each edge then goes from an event that happened further in the past to a more recent one. This timestamp characteristic is sufficient to ensure that there are no cycles.
\end{proof}
\begin{definition}[$\alpha$-nodes] \label{def:alpha-node}
    An $\alpha$-node is a set of vertices $S$ from $G_{\text{\graphName{}}}$ such that, exists a transaction $T$ such that its set of inputs $I_T = S$ and
    \begin{itemize}
        \item[(i)] $S$ is a \texttt{coinbase}\footnote{A coinbase transaction is a special transaction in the Bitcoin protocol creating new coins as mining rewards \cite{Antonopoulos2017}.} transaction, or 
        \item[(ii)] $\exists \, s \in S$ such that $s$ is spending an output with a \texttt{BTC Address}.
    \end{itemize}
\end{definition}
\begin{remark} \label{remark:inputs-are-alpha}
    For each transaction $T$, its set of inputs $I_T$ fulfills:
    \begin{itemize}
        \item[(i)] $I_T$ is an $\alpha$-node, or
        \item[(ii)] $\forall s \in I_T$, $s$ is spending an output with a \texttt{None} address.
    \end{itemize}
\end{remark}
The introduction of $\alpha$-nodes and the previous remark identifies a natural contracted graph of $G_{\text{\graphName{}}}$.
\begin{definition}[Contracted \graphName{} graph]
    The \textbf{Contracted \graphName{} graph}, $G_{\text{\graphName{}}}^*$, is the graph resulting of applying the following two transformations to $G_{\text{\graphName{}}}$:
    \begin{itemize}
        \item[(i)] Identify (contract) all vertices~\cite{Pemmaraju2003} in an $\alpha$-node. Repeat for each different $\alpha$-node contained in $G_{\text{\graphName{}}}$.
        \item[(ii)] For each transaction \textbf{T} fulfilling the second condition in Remark~\ref{remark:inputs-are-alpha},
            \subitem (a) for each spending output \textbf{o} of each input in \textbf{$I_T$}, add an edge from \textbf{o} to each output in \textbf{T}.
            \subitem (b) remove every vertex in \textbf{$I_T$}, as well as its inbound and outbound edges.
    \end{itemize}
\end{definition}
\begin{remark} \label{remark:still-a-dag}
    The transformations applied to $G_{\text{\graphName{}}}$ do not introduce cycles and, as a consequence, $G_{\text{\graphName{}}}^{\ast}$ is also a DAG.
\end{remark}
To define the subgraphs in the TIO graph relevant for our research, we still have to introduce some more concepts.
\begin{definition}[Termination application] \label{def:termination-application}
    Let $f$ be a function,
    $ f: \, \, <\text{\graphName{}}> \longrightarrow \lbrace 0,1 \rbrace $
    defined as follows:
    $$f(x) = \left\lbrace \begin{array}{cl} 0 & \text{ if $x$'s address is \texttt{None}} \\ 1 & \text{otherwise} \end{array} \right. $$
\end{definition}
\noindent
    Given a weakly connected single-source DAG, we call all nodes that are not the source nor sinks \emph{internal nodes}.
\begin{definition}[Unknown TX graph] \label{def:unknown-graph}
    An \textbf{Unknown TX graph} is a single-source, weakly connected, maximal induced subgraph of $G_{\text{\graphName{}}}^*$ such that:
    \begin{itemize}
        \item[(i)] The source $s$ is an $\alpha$-node.
        \item[(ii)] Each sink $t$ fulfills $f(t) = 1$.
        \item[(iii)] Each internal node $v$ fulfills $f(v) = 0$.
    \end{itemize}
\end{definition}
Unknown TX graphs will be our object of study for the rest of the paper. From their construction, we observe the following points.
\begin{definition}[T-DAG] \label{def:t-dag}
    A \textbf{T-DAG} is a single-source directed unlabeled acyclic weakly connected graph.
\end{definition}
\begin{remark}
    Lemma~\ref{lemma:is-a-dag} and Remark~\ref{remark:still-a-dag} prove that an Unknown TX graph is a T-DAG.
\end{remark}
From now on, we will refer to Unknown TX graphs as \textbf{Unknown TX T-DAGs}\footnote{Unlike trees, a vertex in a T-DAG may have more than one parent.}.
\begin{remark}
    If we fix a source $s$, then there exists only one Unknown TX T-DAG with $s$ as its root. We can then denote as $G(s)$ the Unknown TX DAG generated by a given root $s$.
\end{remark}
\begin{definition}[Set of Unknown TX T-DAGs] \label{def:set-dags}
    We define the set of Unknown TX T-DAGs, $\mathcal{D}$, as follows:
\begin{equation*}
    \begin{split}
        \mathcal{D} = & \left\lbrace G(s) : s \text{ is an } \alpha\text{-node and $G(s)$ has at least two vertices} \right\rbrace
    \end{split}
\end{equation*}
\end{definition}
Algorithm~\ref{alg:dag-generation} presents a procedure to generate the Unknown TX T-DAG given an $\alpha$-node $s$.
\begin{algorithm}
  \caption{Unknown TX T-DAG Generation from its root. \label{alg:dag-generation}}
  \begin{algorithmic}[1]
    \Procedure{T-DAG Generation}{$s$}
    \State $G \gets \text{Graph}()$
    \State $S \gets \text{Stack}()$
        \State $G$.addNode($s$)
        \ForAll{$tx\_out$ \textbf{in} getFundOutput($s$)}
            \State $G$.addNode($tx\_out$)
            \State $G$.addEdge($s$, $tx\_out$)
            \State $S$.push($tx\_out$)
        \EndFor
        \While{\textbf{!} $S$.isEmpty()}
            \State $tx\_out \gets S.$pop()
            \If{\textbf{!} $f (tx\_out)$}
                \State $tx\_in \gets$ getFundInput($tx\_out$)
                \ForAll{$new\_out$ \textbf{in} getFundOutput($tx\_in$)}
                    \State $G$.addNode($new\_out$)
                    \State $G$.addEdge($tx\_out$, $new\_out$)
                    \State $S$.push($new\_out$)
                \EndFor
            \EndIf
        \EndWhile
    \EndProcedure
  \end{algorithmic}
\end{algorithm}
An example of an Unknown TX DAG is presented in Figure~\ref{fig:unk-tx-DAG}. Note that, we attach the associated address for each node.

\begin{figure}[!ht]
    \centering
    \includegraphics[width=0.95\columnwidth]{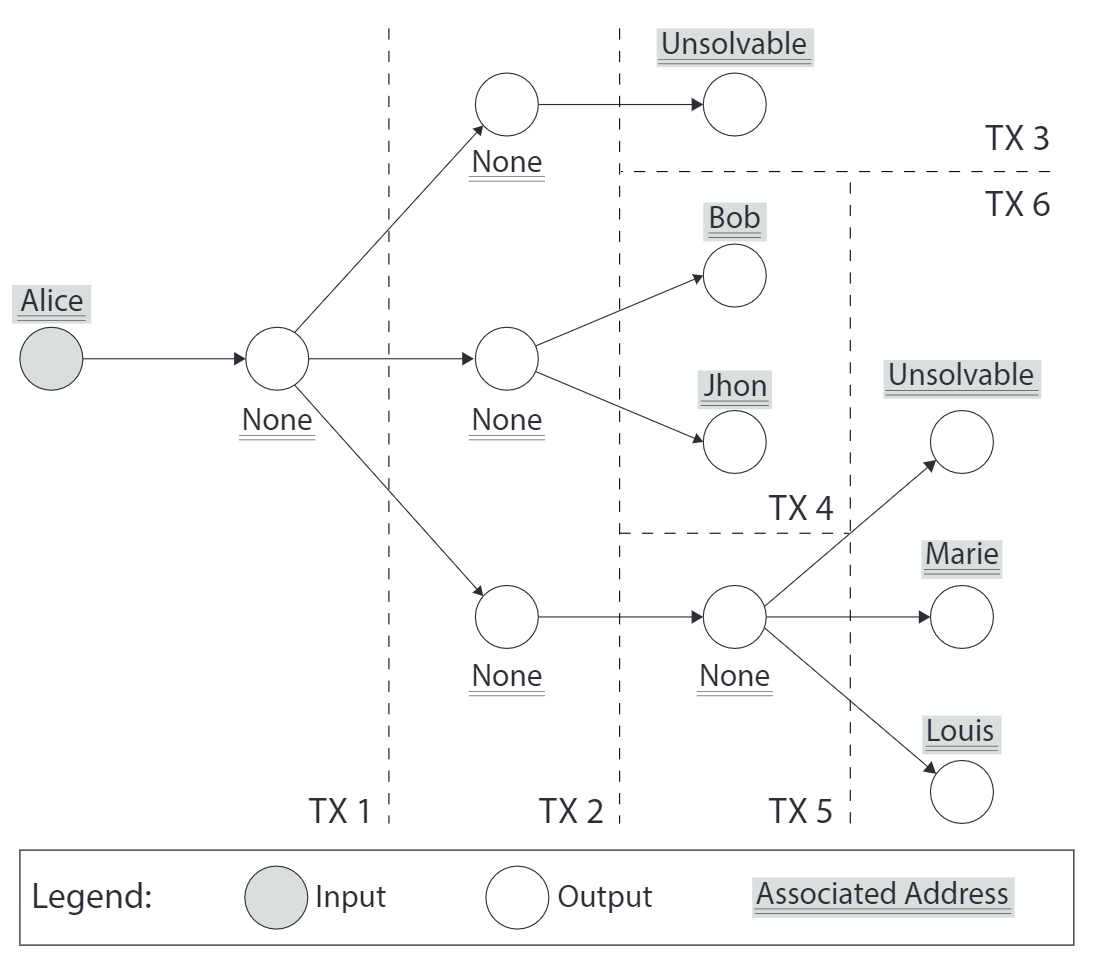}
    \caption{Representation of an Unknown TX T-DAG.\label{fig:unk-tx-DAG}}
\end{figure}

\subsubsection*{T-DAG Compression}
Before any further processing of the Unknown TX T-DAGs, we are interested in pruning the inner nodes that were published to the blockchain without a locking mechanism. These are \emph{trivial} locking scripts that can be solved by anyone at any moment. 
It is easy to verify that a pruned maximal unknown TX T-DAG is still a T-DAG whose leaves fulfill the termination condition. The pruning process involves removing the trivial output and linking the funded outputs with the spending ones. Figure~\ref{fig:compression} illustrates the compression process~\footnote{Note that, since inputs are initially removed from the compressed graph, we actually link funded outputs with the spending outputs associated to the removed inputs. Nonetheless, in Figure~\ref{fig:compression}, we have included the input for better comprehension.}.

\begin{figure}[!ht]
    \centering
    \includegraphics[width=\columnwidth]{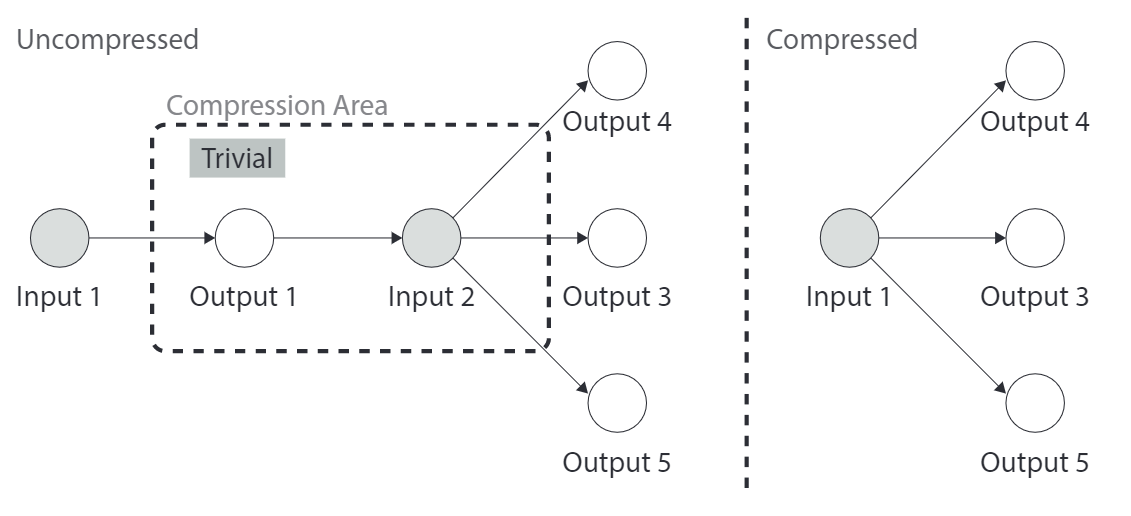}
    \caption{Illustration of a compression process.\label{fig:compression}}
\end{figure}

From now on, when we refer to an Unknown TX DAG we will actually be referring to its compressed version.

\subsection{Chain Abstractions and Post-Processing} \label{subsec:gwflow-4}
After compressing an Unknown TX DAG, we need to store it. To do so, we will use a canonical representation of our graph. A \textit{canonical labeling} is a string derived from a graph such that two different graphs are isomorphic if and only if they yield the same string. A \textit{canonical representation} (i.e. the canonical labeling associated to a set of isomorphic graphs) is then the representative of a graph isomorphism class. In this section we present a total ordering for T-DAG isomorphism classes and a canonical labeling for T-DAGs. 

\subsubsection*{A total ordering for T-DAG isomorphism classes}
Let us first introduce the notation that we will use in this section. Let $T$ be a T-DAG (see Definition~\ref{def:t-dag}) with its root denoted by $t$. Given a T-DAG $T = (V, E)$, we say that $|T| \coloneqq |V|$. 
Given a vertex $v \in V$, $\Gamma^+(v)$ is its outdegree. 
Lastly, given a T-DAG $T$ the children of the root $t$ are $t_1, \dots, t_{\Gamma^+(t)}$. The set $(T_1, \dots , T_{\Gamma^+(t)})$ denotes the maximal collection of sub DAGs induced on $T$ having $t_1, \dots, t_{\Gamma^+(t)}$ as roots.
\begin{definition}[$\prec$ - relation] \label{def:relation}
    Given two T-DAGs, $S$ rooted at $s$ and $T$ rooted at $t$, we say $S \prec T$ if
    \begin{itemize}
        \item[(i)] $|S| < |T|$, or
        \item[(ii)] $|S| = |T| \wedge \Gamma^+(s) < \Gamma^+(t)$, or
        \item[(iii)] $|S| = |T| \wedge \Gamma^+(s) = \Gamma^+(t) = k \, \wedge$ for the first index $i \leq k$ for the ordered sets $(S_1, \dots, S_k)$ and $(T_1, \dots, T_k)$ (where $S_1 \preceq \dots \preceq S_k$ and $T_1 \preceq \dots \preceq T_k$) where $S_i$ differs from $T_i$, it holds $S_i \prec T_i$.
    \end{itemize}
\end{definition}
\begin{definition}[$\equiv$ - equality] \label{def:equality}
    Given two T-DAGs $T$ and $S$, we say $T \equiv S$ if neither $T \prec S$ nor $T \succ S$ hold.
\end{definition}
Let $\cong$ be the isomorphism operator. We derive the following lemma: 
\begin{lemma} \label{lemma:equiv-iso}
    Given two T-DAGs $T$ and $S$ with $|T| = |S| = n$, then $ T \equiv S \Leftrightarrow T \cong S $.
\end{lemma}

\begin{proof}
    We prove each implication separately,

        $[\Rightarrow ]$  Let $S_n$ be the symmetric group acting on the vertices of $T$, $V(T)$. Given $\sigma \in S_n$, we denote the action of $\sigma$ on $v \in V(T)$ by $\sigma(v)$. We can naturally extend the definition to sets of vertices, $S \subseteq V, \sigma (S) = \lbrace \sigma(v) : v \in S \rbrace$, and to the T-DAG itself $\sigma(T) \coloneqq (\sigma(V), E')$, where $E' = \lbrace (\sigma(u), \sigma(v)) : (u,v) \in E(T) \rbrace$. \begin{remark} \label{remark:permutation} Let $S$ and $T$ be two T-DAGs, $|T| = |S| = n$. Then, $$S \cong T \Leftrightarrow \exists \, \sigma \in S_{n} \text{ such that } \sigma(S) = T$$ \end{remark} We now prove this direction of the lemma by induction on the size of the T-DAG, $|T|$.
        
        \textbf{If $|T| = 1$:} then both $S$ and $T$ are T-DAGs formed by a single vertex, hence they are the same graph and therefore isomorphic taking the identity permutation.
        
        \textbf{If $|T| = n$:} to prove the inductive step we assume that for any pair of T-DAGs with size $< n$, then $T \equiv S \Rightarrow T \cong S$. If now $|T| = n$, $T \equiv S \Rightarrow |T| = |S| = n$, $\Gamma^+(t) = \Gamma^+(s) = k$, and $(T_1, \dots, T_k) \equiv (S_1, \dots, S_k)$ pairwise, where $S_1 \preceq \dots \preceq S_k$ and $T_1 \preceq \dots \preceq T_k$. That is, $\forall i \in {1, \dots k} \left\lbrace \begin{array}{c} T_i \equiv S_i \\ |T_i| = |S_i| < n \end{array} \right. \overset{\tiny{Ind. H}}{\Longrightarrow} T_i \cong S_i \overset{\tiny{R.~\ref{remark:permutation}}}{\Longrightarrow} \exists \sigma_i \in S_{|T_i|} \text{ such that } \sigma_i (T_i) = S_i$. We now consider the following permutation: $\sigma = \sigma_1 \circ \dots \circ \sigma_k \circ (t \rightarrow s)$, the composition of all the permutations that match each subtree and the map from a root to the other. $\sigma$ fulfils that $\sigma(T) = S \overset{\tiny{R.~\ref{remark:permutation}}}{\Longrightarrow} T \cong S$
        
        $[\Leftarrow ]$ We argue again by induction on the size of the T-DAG. The base case is the same as before so we do not repeat it. For the induction step, we have:
        
        \textbf{If $|T| = n$:} $T \cong S \Rightarrow |T| = |S| = n \wedge \Gamma^+(t) = \Gamma^+(s) = k$. Additionally, ordering the subtrees $(T_1, \dots, T_k)$, $(S_1, \dots, S_k)$ such that $T_1 \preceq \dots \preceq T_k$ and $S_1 \preceq \dots \preceq S_k$ necessarily $T_i \cong S_i \forall i \in \lbrace 1, \dots, n \rbrace$ with $|T_i| = |S_i| < n \overset{Ind. H}{\Longrightarrow} T_i \equiv S_i \Rightarrow (T_1, \dots, T_k) \equiv (S_1, \dots, S_k) \Rightarrow T \equiv S$. \\
\end{proof}

\begin{remark} \label{remark:is-total-ordering}
    The operators $(\prec, \succ, \equiv)$ induce a total ordering on T-DAGs isomorphism classes.
\end{remark}
\begin{proof}
    Given $T$ and $S$ two T-DAGs,
    \begin{itemize}
        \item[(i)] \textbf{Antisymmetry:} $S \preceq T \wedge S \succeq T \Leftrightarrow \neg (S \succ T) \wedge \neg (S \prec T) \Leftrightarrow S \equiv T \Leftrightarrow S \cong T$
        \item[(ii)] \textbf{Transitivity:} clearly holds by definition.
        \item[(iii)] \textbf{Connex Property:} $S \preceq T \vee S \succeq T \Leftrightarrow \neg (A \succ B) \vee \neg (A \prec B) \Leftrightarrow \neg (A \succ B \wedge A \prec B) \Leftrightarrow \neg 0 = 1$
    \end{itemize}
\end{proof}
\subsubsection*{Canonical labeling for T-DAGs}
We now introduce a canonical labeling for T-DAGs.  We provide unique representatives for T-DAG isomorphism classes and their string representation. 

\begin{definition}[$\Delta$-operator]
    The \textbf{indegree operator} ($\Delta$) is a total ordering on equivalence classes of the $\equiv$-relation. Let $T$ be a T-DAG such that $T_1 \preceq \dots \preceq T_{\Gamma^+(t)}$. That is, it takes a set of representatives of the $\equiv$-relation and orders it. Formally,
    \begin{equation*}
        \begin{split}
            \Delta : & \lbrace T_1, \dots, T_{\Gamma^+(t)} \rbrace / \equiv \longrightarrow \lbrace T_1, \dots, T_{\Gamma^+(t)} \rbrace / \equiv \\
            & \lbrace \overline{T_i}, \dots, \overline{T_{i+k}} \rbrace \longmapsto \Delta ( \lbrace \bar{T_i}, \dots, \overline{T_{i+k}} \rbrace) \coloneqq (\lbrace \overline{T_i}, \dots, \overline{T_{i+k}} \rbrace, \leq_{\ast})
        \end{split}
    \end{equation*}
    where $k \in \mathbb{N}$ and $(\lbrace \overline{T_i}, \dots, \overline{T_{i+k}} \rbrace, \leq_{\ast})$ is the totally ordered set according to the following relation:
    \begin{equation*}
        \begin{split}
            \overline{T_i} \leq_{*} \overline{T_j}  \Leftrightarrow & \left( \lbrace |\Gamma^-(t_{i_1})|, \dots, |\Gamma^-(t_{i_{\Gamma^+(t_i)}})| \rbrace, \leq \right)  \leq \\ & \left( \lbrace |\Gamma^-(t_{j_1})|, \dots, |\Gamma^-(t_{j_{\Gamma^+(t_j)}})| \rbrace, \leq \right)
        \end{split}
    \end{equation*}
    That is, $\overline{T_i} \leq_{\ast} \overline{T_j}$ iff the non-decreasing indegree sequence of $t_i$'s children is pairwise smaller than that of $t_j$.
\end{definition}
\noindent
Naturally, applying the operator to the whole set (i.e. $\Delta(T)$) means applying it element-wise in the quotient set, reordering only elements that were considered $\equiv$-equal.
\begin{definition}[T-DAG isomorphism classes representative] \label{def:t-dag-canonical}
    Given a T-DAG, $T$, we reorder it so that $T_1 \preceq \dots \preceq T_{\Gamma^+(t)}$. We denote the reordered T-DAG with $T^{\ast}$. The representative of $T$'s isomorphism class $\bar{T}$ is defined as $\bar{T} \coloneqq \Delta(T^{\ast})$.
\end{definition}
\begin{lemma} \label{lemma:repr-well-defined}
    Given $T$ and $S$ T-DAGs, $T \cong S \Rightarrow \bar{T} = \bar{S}$. Thus $\bar{T}$ is well defined.
\end{lemma}
\begin{proof}
    Given a vertex $v \in V(T)$, the \textbf{height} of $v$ is the number of edges of the longest path between $v$ and one of its leafs. Let $V(T)_h \subset V(T),$ be the set of vertices with height equal to $h$. We prove that, given an $h$, the set of maximal induced T-DAGs rooted at $V(T)_h$ and $V(S)_h$, reordered with $\preceq$ and then with $\Delta$, are the same. In particular, when $h$ equals the height of the T-DAG (i.e. the height of its root), this yields $\bar{T} = \bar{S}$. We proceed by induction on the height $h$.

    \textbf{If the height is 0:} $T \cong S \Rightarrow | \lbrace v \in V(T) : \Gamma^+(v) = 0 \rbrace | = | \lbrace v \in V(S) : \Gamma^+(v) = 0 \rbrace | = k$. In fact both T-DAGs have the same number of leafs and therefore $\lbrace t_1, \dots, t_k \rbrace = \lbrace s_1, \dots, s_k \rbrace$.

    \textbf{If the height is $h$:} we assume that, for heights $ \leq h$, the set of T-DAGs reordered with $\preceq$ and then with the $\Delta$ operator are the same. $T \cong S \Rightarrow | V(T)_{h+1} | = | \lbrace v \in V(T) : \text{height}(v) = h + 1 \rbrace | = | \lbrace v \in V(S) : \text{height}(v) = h + 1 \rbrace | = | V(S)_{h+1} | = k$. We now order $V(T)_{h+1}$ and $V(S)_{h+1}$ in non-decreasing outdegree order and we apply the $\Delta$ operator to $\lbrace T_1, \dots, T_k \rbrace$ and $\lbrace S_1, \dots, S_k \rbrace$, the T-DAGs with roots in $V(T)_{h+1}$ and $V(S)_{h+1}$. We will refer to the before-mentioned roots as $\lbrace t_1, \dots, t_k \rbrace$ and $\lbrace s_1, \dots, s_k \rbrace$, and to the $j-$th sibling of the $i-$th root as $t_i^j$, where $j \in \lbrace 1, \dots, \Gamma^+(t_i) \rbrace$ (with $T_i^j$ being the T-DAGs rooted at this nodes). $\lbrace T_1, \dots, T_k \rbrace$ and $\lbrace S_1, \dots, S_k \rbrace$ ordered in this manner satisfy $T_{1} \preceq \dots \preceq T_{k} $ and $S_{1} \preceq \dots \preceq S_{k}$. Furthermore, for each T-DAG $T_i$, $i \in \lbrace 1, \dots, k \rbrace$, with root $t_i$ at layer with height $h$, we have that
    \begin{equation*}
        \begin{split}
            \left.
            \begin{array}{l}
                \Gamma^+(t_i) = \Gamma^+(s_i) \\
                \forall j \in \lbrace 1, \dots, \Gamma^+(t_i) \rbrace \left\lbrace \begin{array}{c} \Gamma^-(t_i^j) = \Gamma^-(s_i^j) \\ \text{Ind. H} \Rightarrow  T_i^j = S_i^j  \end{array} \right.
            \end{array}
            \right\rbrace \Rightarrow T_i = S_i
        \end{split}
    \end{equation*}
    All vertices with height $h+1$ have as children the roots of T-DAGs with heights $\leq h$. 
    
    If now we make $h+1$ equal $T$ and $S$'s height, we have $V(T)_{h+1} = t$ and $V(S)_{h+1} = s$. Therefore, the set of maximal induced T-DAGS are $\lbrace T \rbrace$ and $\lbrace S \rbrace$ respectively. We have proven that, reordering with $\preceq$ and $\Delta$, both sets are equal. Hence, $\bar{T} = \bar{S}$.
\end{proof}
\begin{remark} \label{remark:t-cong-bart}
    It holds $T \cong \bar{T}$.
\end{remark}
\begin{proof}
    From Definition~\ref{def:t-dag-canonical} we observe that, in order to obtain $\bar{T}$ from $T$, we reorder the subtrees non-decreasingly and we apply the $\Delta$ operator. Let then $\sigma$ be a permutation such that $\sigma(T)$ generates $T_1 \preceq \dots \preceq T_{\Gamma^+(t)}$. We then consider $\mu$ as the permutation resulting of doing $\sigma$ and $\Delta$ one after the other in this order. From Remark~\ref{remark:permutation} it follows that $\mu (T) = \bar{T} \overset{R.~\ref{remark:permutation}}{\Longrightarrow} T \cong \bar{T}$.
\end{proof}

\begin{definition}[T-DAG labeling] \label{def:t-dag-labeling}
    Given a T-DAG $T$ we identify its vertices traversing $T$ breadth-first with a FIFO queue and, starting from the root, for each new vertex (not identified that we dequeue) we assign it the current vertex count value, increment the count by one and queue its set of children. 
    The \textbf{labeling} lbl$(T)$, associated to $T$, is the string result from traversing the identified $T$ breadth-first with a FIFO queue and, starting from the root, for each vertex (not processed that we dequeue) we append each of its children identifier to the labeling and queue each of its children. We separate children of the same parent with a comma ',', sets of siblings with a colon ':', and we denote the end of the label with a semi-colon ';'~\footnote{We are aware that these separators depend heavily on the labeling implementation.}. We will refer to the identifier (the label) of a vertex $v$ obtained through this procedure as id$(v)$.
\end{definition}
\begin{definition}[T-DAG canonical labeling] \label{def:t-dag-canonical-labeling}
    Given a T-DAG $T$, the \textbf{canonical labeling} of $T$, c$(T)$ is the \textbf{labeling} of its isomorphism class representative, $\bar{T}$. That is, c$(T) \coloneqq \text{lbl}(\bar{T})$.
\end{definition}
In a nutshell, the canonical labeling is obtained with a total ordering for T-DAGs, an indegree-based operator and an additional labeling 
Sufficiency is proven in Lemma~\ref{lemma:canonical-label-well-defined}.

Before proving the sufficiency of the three operations, we introduce additional concepts regarding labelings.
\begin{definition}[Labeling clause]
    Given a labeling lbl$(T)$ of a T-DAG $T$, a \textbf{clause} in the labeling is a set of identifiers contained within either two colons, a colon and a semi-colon, or the beginning of the labeling and a colon.
\end{definition}
\noindent
Note that, given a labeling lbl$(T)$, we can index the clauses by order of appearance starting from $0$. Thus, we can think of lbl$(T)$ as an ordered array of clauses: lbl$(T) = [c_0, \dots, c_n]$, where $n = |V(T)|$. 
Thus, an upper-bound to get the length of each clause, by preprocessing the label, is $\mathcal{O}(m)$, where $m = |E(T)|$.

Given a T-DAG labeling lbl$(T)$, we can obtain an array \texttt{out\_deg} that in the $i$-th position contains the outdegree of $v \in V(T)$ such that id$(v) = i$. Algorithm \ref{alg:outdeg-from-labeling} presents a pseudo-code to do so.
The proofs of correctness (Lemma \ref{lemma:outdegree-sequence}) and complexity (Lemma \ref{lemma:outdegree-sequence-complexity}) are omitted for space reasons.


\begin{algorithm}
  \caption{Outdegree sequence from a label. \label{alg:outdeg-from-labeling}}
  \begin{algorithmic}[1]
    \Procedure{Outdegree Parsing}{lbl$(T)$}
        \State $m \gets \max (\text{lbl}(T))$ \label{line:m}
        \State $out\_deg \gets \text{zeros-array}(m+1)$ \label{line:test}
        \State $processed \gets \text{zeros-array}(m+1)$
        \State $Q \gets \text{FIFO\_Queue}()$
        \State $n\_id \gets 0$
        \For{$i: 1$ \textbf{to} $\#$ clauses in lbl($T$)} \label{line:for}
            \While{\textbf{!} $Q.\text{empty}()$} \label{line:while}
                \State $n\_id \gets Q.\text{dequeue}()$
                \If{\textbf{!} $processed[n\_id]$}
                    \State \textbf{break}
                \EndIf
            \EndWhile
            \If{\textbf{!} $processed[n\_id]$}
                \State $processed[n\_id] \gets \text{True}$ \label{line:process}
                \State $out\_deg[n\_id] \gets \text{len}(c_i)$ \label{line:assign}
                \ForAll{$id$ \textbf{in} $c_i$}
                    \State $Q$.queue($id$) \label{line:push}
                \EndFor
            \EndIf
        \EndFor
        \State \Return $out\_deg$
    \EndProcedure
  \end{algorithmic}
\end{algorithm}
\begin{lemma} \label{lemma:outdegree-sequence}

    Algorithm \ref{alg:outdeg-from-labeling} is correct.
\end{lemma}
\begin{lemma} \label{lemma:outdegree-sequence-complexity}
    The Algorithm described in Algorithm~\ref{alg:outdeg-from-labeling} is linear in the number of edges of the labeled T-DAG $T$. Thus, if $m = |E(T)|$, the complexity is $\mathcal{O} (m)$.
\end{lemma}
 \begin{lemma} \label{lemma:exist-diferent-outdegree}
        Let $S$ and $T$ be T-DAGs with $|S| = |T|$ and $\Gamma^+(s) = \Gamma^+(t) = k$. Let $\bar{S}$ and $\bar{T}$ be their isomorphism class representatives.
        $$ \bar{S} \ncong \bar{T} \Rightarrow \exists \, u \in V(\bar{S}), v \in V(\bar{T}) \text{ s.t. } $$
        $$ \text{id}(u) = \text{id}(v) \text{ but } \Gamma^+(u) \neq \Gamma^+(v) $$
        where id is the identifier assigned to each vertex when labeled.
    \end{lemma}
    
    \begin{proof}
        W.l.o.g. we assume $\bar{S} \prec \bar{T}$. This implies that either $(i)$, $(ii)$, or $(iii)$ from Definition~\ref{def:relation} must hold. Let us define the following set of indices, $D \coloneqq \lbrace i \in \lbrace 1, \dots, k \rbrace : \bar{S_i} \ncong \bar{T_i} \rbrace $. From the hypothesis we can ensure that $D \neq \emptyset$. For each index $i \in D$, we take $u_i$ the first vertex traversing $\bar{S_i}$
        
        breadth-first with a FIFO queue such that condition $(ii)$  fails. Note that, from Definition~\ref{def:equality} and Lemma~\ref{lemma:equiv-iso}, $\bar{S_i} \ncong \bar{T_i}$ implies that at some point of the recursion $(i)$ or $(ii)$ will fail. And if $(i)$ fails, necessarily does $(ii)$. Let $U$ be, $U \coloneqq \lbrace u_i : i \in D \rbrace$, and $u^{\ast}$,
        $$ u^{\ast} \coloneqq \underset{u \in U}{\text{argmin}} \lbrace id(u) \rbrace $$
    Let $v^{\ast}$ be the corresponding vertex in $V(\bar{T})$ that made condition $(ii)$ fail for $u^{\ast}$. For all vertices with a smaller id, the outdegree is the same. Additionally, since $\bar{S}$ and $\bar{T}$ are both ordered non-decreasingly, when traversed breadth-first with a FIFO queue, we can affirm that $u^{\ast}$ and $v^{\ast}$ will be given the same identifier but they have different outdegrees.
\end{proof}
\begin{lemma} \label{lemma:canonical-label-well-defined}
    Given two T-DAGs $T$ and $S$, it holds $ T \cong S \Leftrightarrow c(T) = c(S) $. Thus, the canonical labeling presented in Definition \ref{def:t-dag-canonical-labeling} is well defined.
\end{lemma}

\begin{figure*}[]
    \centering
    \includegraphics[width=0.95\textwidth]{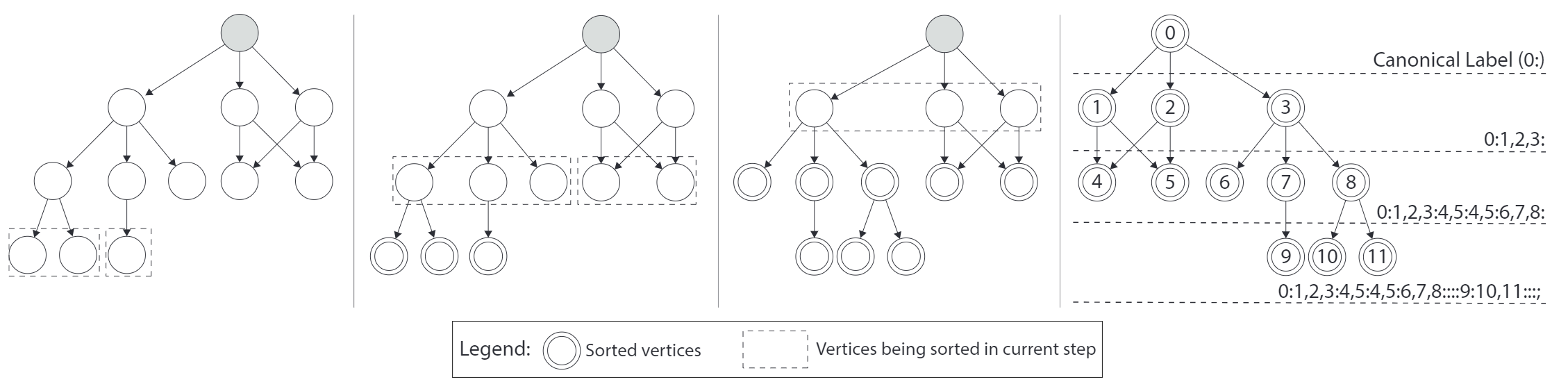}
    \caption{From left to right, obtaining the isomorphism class representative and the canonical labeling.\label{fig:obtain-class-repr}}
\end{figure*}

\begin{proof}
    We prove each implication separately,

    $[\Rightarrow]$  This implication is an immediate consequence of Lemma \ref{lemma:repr-well-defined}. $T \cong S \Rightarrow \bar{T} = \bar{S} \Rightarrow c(T) = c(S)$, since the process described in Definition \ref{def:t-dag-canonical-labeling} is deterministic.
        
    $[\Leftarrow]$ We will prove this implication by contrapositive. We will see that $T \ncong S \Rightarrow c(T) \neq c(S)$. Consider two T-DAGs, $T$ and $S$, such that $T \ncong S$.  
        
    Since $T \ncong S$ then it must be $\bar{T} \ncong \bar{S}$ what implies $\bar{T} \not\equiv \bar{S}$. Therefore, either $\bar{T} \prec \bar{S}$ or $\bar{S} \prec \bar{T}$ must hold. Without loss of generality we assume $\bar{S} \prec \bar{T}$. Lemma~\ref{lemma:outdegree-sequence} proofs that, given a T-DAG labeling, lbl$(T)$, we can obtain $|T|$ and $T$'s outdegree sequence. As a consequence, lbl$(\bar{S}) = $ lbl$(\bar{T}) \Rightarrow |\bar{S}| = |\bar{T}|$, and two vertices with the same id have the same outdegree. Furthermore, from Definition~\ref{def:relation} we recall:
         
    $$ \bar{S} \prec \bar{T} \Rightarrow \left\lbrace \begin{array}{l} (i) |\bar{S}| < |\bar{T}| \text{ , or} \\ (ii) |\bar{S}| = |\bar{T}| \wedge \Gamma^+(\bar{s}) < \Gamma^+(\bar{t}) \text{ , or} \\ (iii) |\bar{S}| = |\bar{T}| \wedge \Gamma^+(\bar{s}) = \Gamma^+(\bar{t}) = k \wedge \\ \hspace{19pt} \left( \bar{S_1}, \dots, \bar{S_k} \right) \prec \left( \bar{T_1}, \dots, \bar{T_k} \right) \end{array} \right. $$
    
    $(i) \Rightarrow |\bar{S}| \neq |\bar{T}| \overset{L.~\ref{lemma:outdegree-sequence}}{\Longrightarrow} \text{ lbl}(\bar{S}) \neq \text{ lbl}(\bar{T}) \Rightarrow \text{ c}(S) \neq \text{ c}(T)$
    
    $(ii) \Rightarrow\Gamma^+(\bar{s}) \neq \Gamma^+(\bar{t}) \Rightarrow $ clause $\# 1$ in each label will have a different size $\Rightarrow \text{ c}(S) \neq \text{ c}(T)$

    If condition $(i)$ or $(ii)$ happen, we are done with the proof. Otherwise we can assume $|\bar{S}| = |\bar{T}|$ and $\Gamma^+(\bar{s}) = \Gamma^+(\bar{t}) = k$. 

From Lemma~\ref{lemma:outdegree-sequence} we state the following:
        
    $$\text{lbl}(\bar{S}) = \text{ lbl}(\bar{T}) \Rightarrow \forall v \in V(\bar{S}), u \in V(\bar{T}); $$ 
    $$ \text{ id}(v) = \text{ id}(u) \Rightarrow \Gamma^+(v) = \Gamma^+(u) $$ \\[-10pt]
    That is, if two vertices have the same identifier (and the T-DAGs they belong to have the same label), they must have the same outdegree.
    
    Equivalently, if this does not hold, then two T-DAGs cannot have the same labeling. Formally, 
    $$ \exists v \in V(\bar{S}), \exists u \in V(\bar{T}) \text{ s.t. } \text{ id}(v) = \text{ id}(u) \text{ but } $$
    $$ \Gamma^+(v) \neq \Gamma^+(u)  \Rightarrow \text{lbl}(\bar{S}) \neq \text{lbl}(\bar{T}) $$ 
    Then, if we assume $(iii)$ holds, 
    
    $(iii) \Rightarrow |\bar{S}| = |\bar{T}| \wedge \Gamma^+(\bar{s}) = \Gamma^+(\bar{t}) = k$ we also have $\bar{S} \ncong \bar{T} \overset{L.~\ref{lemma:exist-diferent-outdegree}, L.~\ref{lemma:outdegree-sequence}}{\Longrightarrow} \text{ lbl}(\bar{T}) \neq \text{ lbl}(\bar{S}) \Leftrightarrow \text{ c}(S) \neq \text{ c}(T)$
    
    And we have proven that $T \ncong S \Rightarrow c(T) \neq c(S)$. Thus,  $c(T) = c(S) \Rightarrow T \cong S$.
\end{proof}
To obtain the canonical labeling from any given T-DAG, we first obtain its isomorphism class representative and then the labeling induced from it. Figure~\ref{fig:obtain-class-repr} presents a bottom up strategy to obtain the class representative and the associated canonical labeling (it should be read from left to right, top to bottom).
In the first four graphs, we perform a bottom up approach to reorder the T-DAG. Bold vertices are those already sorted and the dashed boxes mark which sets of siblings are to be sorted next. Finally, given the class representative ($\bar{T}$), we provide the canonical labeling.

Using the constructions and the efficient labeling derived in this section, we can state our main theoretical result.
\begin{theorem}\label{mainthm}
The set of \unknown{} inputs/outputs induces a family of subgraphs (patterns), within the Bitcoin User Network, which can be efficiently labeled and tested for isomorphism.
\end{theorem}
To the best of our knowledge, this is the first result that systematically addresses \unknown{} transactions, which are often neglected in the current technical literature. 
In the next section, we will use our isomorphism algorithm to cluster TX T-DAGs that interact with the same patterns. Note that, while presented only for \unknown{} transactions, our approach immediately extends to any transaction system where patterns can be modeled by T-DAGs.

\section{Experimental Results}\label{sec:implementation}
In this section, we describe the methodology adopted to implement and test our model on the Bitcoin ledger and discuss the achieved results. Our approach consists of multiple steps. First, we parsed the Bitcoin blockchain according to the framework introduced in Section~\ref{subsec:gwflow-1}. Blocks and transactions information were retrieved by querying directly a Bitcoin full-node, importing data into a MySQL database. 
We then retrieved all the \unknown{} transactions with a valid locking script from our database, and we applied our methodology to build TX T-DAGs and study their patterns.
Finally, we clustered all the TX T-DAGs found in the previous step according to their isomorphism classes \ndssColorTex{to group all similar patters that could have been created by the same entity}.
The above steps are described in more details in the following subsections.

\subsection{Database}\label{subsec:database}
Complete and reliable Bitcoin blockchain data are essential to correctly build the TX T-DAGs. The official software release used by the Bitcoin protocol, i.e. Bitcoincore, 
is not suitable for this purpose as it is not designed for the analysis of the data contained in the blockchain. To the best of our knowledge, all freely available blockchain explorer tools suffer from the problems described in Section~\ref{subsec:gwflow-1}. That is, they do not properly handle custom transactions. For this reason, we have designed a MySQL database to store transaction data retrieved by parsing the Bitcoin ledger using our model discussed in Section~\ref{subsec:gwflow-1}. 
We imported all the Bitcoin blockchain data starting from the genesis block 0 (mined on 2009-01-03), up to block 591,872 (mined on 2019-08-26). Our database consists of more than 448 million Bitcoin transactions, over 1.1 billion transaction inputs, over 1.2 billion transaction outputs and around 550 million Bitcoin addresses. We hosted our MySQL database on a Dell Poweredge R740 Server, equipped with 2 CPU Intel® Xeon® Gold 6144 3.5G, RAM 512GB, running Ubuntu Server 18.04.4 LTS. To parse the Bitcoin ledger, we used Bitcoin Core Daemon v.0.18.0.0, running on a Dell XPS laptop equipped with an Intel® Xeon® CPU E3-1505M 2.80GHz, RAM 32GB - OS: Ubuntu 18.04.4 LTS. 

\begin{table}[]
\centering
\begin{tabular}{|c|c|c|c|c|}
\hline
\textbf{\begin{tabular}[c]{@{}c@{}}Number of \\ isomorphic \\ TX T-DAG\end{tabular}} & \textbf{Height} & \textbf{Cardinality} & \textbf{\begin{tabular}[c]{@{}c@{}}Number \\ of edges\end{tabular}} & \textbf{\begin{tabular}[c]{@{}c@{}}Number \\ of roots\end{tabular}} \\ \hline
{\color[HTML]{000000} 29218} & 2 & {\color[HTML]{000000} 3} & 2 & 1 \\ \hline
{\color[HTML]{000000} 607} & 2 & {\color[HTML]{000000} 11} & 14 & 2 \\ \hline
{\color[HTML]{000000} 51} & 2 & {\color[HTML]{000000} 6} & 7 & 1 \\ \hline
{\color[HTML]{000000} 36} & 2 & {\color[HTML]{000000} 6} & 5 & 1 \\ \hline
{\color[HTML]{000000} 32} & 2 & {\color[HTML]{000000} 5} & 4 & 1 \\ \hline
{\color[HTML]{000000} 25} & 2 & {\color[HTML]{000000} 7} & 6 & 1 \\ \hline
{\color[HTML]{000000} 20} & 2 & {\color[HTML]{000000} 7} & 6 & 1 \\ \hline
{\color[HTML]{000000} 14} & 2 & {\color[HTML]{000000} 6} & 6 & 1 \\ \hline
{\color[HTML]{000000} 12} & 2 & {\color[HTML]{000000} 4} & 3 & 1 \\ \hline
{\color[HTML]{000000} 9} & 2 & {\color[HTML]{000000} 6} & 5 & 1 \\ \hline
{\color[HTML]{000000} 9} & 2 & {\color[HTML]{000000} 40002} & 40000 & 20000 \\ \hline
\textbf{...} & \textbf{...} & \textbf{...} & \textbf{...} & \textbf{...} \\ \hline
1 & 2260 & 6878 & 7176 & 10 \\ \hline
1 & 514 & 2073 & 4144 & 7 \\ \hline
1 & 383 & 1568 & 3096 & 13 \\ \hline
1 & 381 & 1059 & 1058 & 5 \\ \hline
1 & 2 & 5 & 4 & 1 \\ \hline
\end{tabular}
\caption{10 most common isomorphism classes and some other more complex patterns.}
\label{tab:iso_classes}
\end{table}

\begin{figure*}[!ht]
    \centering
    \includegraphics[width=0.75\textwidth]{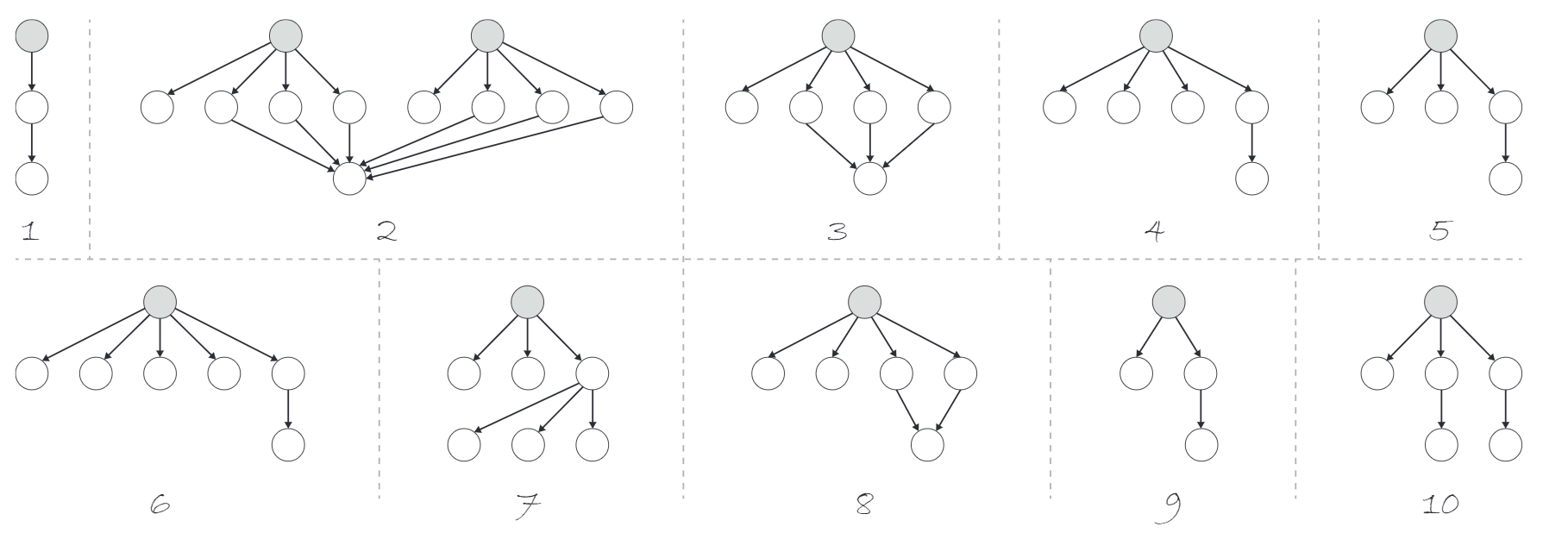}
    \caption{10 most common isomorphism classes.\label{fig:graph-table}}
\end{figure*}

\subsection{Unknown TX T-DAG Construction}\label{subsec:unk_tdag}
\ndssColorTex{After the parsing phase, our database contains all the transactions included in the Bitcoin ledger, both unknown and standard. At this point, we started analyzing the unknown ones, 
identifying around 22 million $\alpha$-nodes that can generate an Unknown TX T-DAG (DAGs where the inner nodes are only unknown transaction outputs).}
We then built a forest by iterating Algorithm~\ref{alg:dag-generation} for each $\alpha$-node.
Each weakly connected component of this forest represents an Unknown TX T-DAG originating from one or more alpha nodes. We used the library networkx 2.3 (together with the Python 3.5 interpreter) to create and manage the forest containing all the Unknown TX T-DAGs.
\subsection{Pruning Phase}
By construction, Unknown TX T-DAGs represent transaction patterns in the blockchain network 
generated by unknown transactions. The root denotes the set of (standard) inputs that generated the pattern. Each inner node represents an unknown transaction output, i.e. an output with a Null value, that has been spent.
Finally, each leaf represents either a known transaction output (with a valid Bitcoin address attached to it) or an unspent transaction output (either with a valid or a null address). Therefore, an unknown TX T-DAG of height 1 is trivial:
in fact, such graph represents a single transaction with an unknown output that has not been spent. An output of this type could be an invalid/unspendable output or simply a custom but valid output that has not been spent in the considered blockchain portion. 
In the first case, since its locking script is malformed, this output is impossible to redeem. Therefore, the corresponding Unknown TX T-DAG can never grow further. In the second case, however, if the unknown output is spent in the future, our methodology will capture this event during the update of our database, and the associated Unknown TX T-DAG structure will be updated and considered for further analysis. Following this observation, we pruned the forest by dropping the weakly connected components of height 1 \arXivColorTex{as they do not represent a relevant pattern for either $\mathcal{T}$ or $\mathcal{U}$}.\\
Finally, we obtained a forest with 803,782 nodes and 797,432 edges, having 30,333 weakly connected components left, i.e., our Unknown TX T-DAGs. These T-DAGs can be easily integrated into the transaction network $\mathcal{T}$ and the user network $\mathcal{U}$. In this way, bot the blockchain analysis tools and the proposals in the literature which use these data structures will rely on complete information, never considered before.\\

\subsection{Isomorphism Detection}\label{subsec:iso_detection}
\ndssColorTex{In the last step of our methodology, we clustered the 30,333 Unknown TX T-DAGs, obtained in the previous phase, according to their isomorphism classes. For each T-DAG, we built its isomorphism class representative (definition~\ref{def:t-dag-canonical}) by using the total ordering introduced in Section~\ref{subsec:gwflow-4}}
Before performing the clustering procedure, we augmented every TX T-DAGs with more than one root, such as the one depicted in Figure~\ref{fig:graph-table}.2. In particular, for each graph of the cited type, we created a new root connected to each of the old ones. 
Consequently, each graph $g$ with multiple roots is converted into a new graph $g^*$ having a single root and the same nodes of $g$ plus one (the new root).
Once all the Unknown TX T-DAGs were standardized to have a single root, we clustered them according to their isomorphism class representative.
We identified 273 different isomorphism classes. Figure~\ref{fig:graph-table} shows the 10 most common classes, while Table~\ref{tab:iso_classes} reports the height, cardinality, number of edges, and number of roots for the same classes and some other ones characterized by a more complex pattern.


\subsection{Discussion}\label{subsec:discussion}
Our  approach is, to the best of our knowledge, the only one capable of identifying transaction inputs and outputs without relying on Bitcoin addresses, providing a framework to correctly handle unknown transactions. 
This feature allowed us to detect unknown transaction patterns not captured by state of the art blockchain analysis techniques. 
The 30,000+ Unknown TX T-DAGs discovered, 
can be easily integrated into the Bitcoin User Network, finally providing a complete transaction history, taking into account also unknown transactions. It is worth highlighting that even the simplest of the identified structures, such as the one in Figure~\ref{fig:graph-table}.1, that appears more than 29,000 times in the Bitcoin ledger, is sufficient to make current parsers not to consider potentially relevant data.
Indeed, although the transaction contains the root (the payer address) and the leaf (the recipient address), the User Network built without considering unknown transactions will not contain the middle node (unknown output), hence breaking the connection between the two users. 
Instead, integrating the User Network with our T-DAGs, automatically increases the accuracy of any Bitcoin clustering heuristic used so far, as well as any deanonymization technique, by simply considering never-used before data.
In addition, our isomorphism classes could lead to new clustering techniques: non-trivial patterns, i.e. transaction schemes not originated from a normal user behavior, can be used to cluster 
services/entities that exhibit the same patterns.
Other than the 10 most common isomorphism classes shown in Figure~\ref{fig:graph-table}, we found several other patterns that deserve particular attention from a semantic point of view.
As an example, we discovered an Unknown TX T-DAG of height 2260, having 6878 nodes, and 7176 edges. This complex pattern has 10 roots, i.e., is generated by 10 different standard transactions and therefore potentially by up to 10 different entities. Started on 2014-02-19, and concluded on 2014-12-08, such a pattern moved a total of 247.36 Bitcoins that, according to the historical Bitcoin prices, were worth about 92,500 US Dollars at the time of the last transaction (December 2014)---a complete analysis of the just introduced graph, and other interesting ones, will be provided in future work.\\
\arXivColorTex{As a further step to refine our results, we have removed T-DAGs generated by trivial transactions, e.g., transactions redeemable by anyone, and T-DAGs generated by non-standard transactions with a known semantic, e.g., transactions generated by software bugs, network tests, and crypto challenges, just to name a few. The goal is to focus only on interesting T-DAGs, i.e., those that may have been created to hide malicious behavior, discarding those generated by known and legitimate activities. In order to perform this filtering, we checked the locking script at the root of each T-DAGs. In particular, we have created regular expressions for each non-standard locking script with trivial semantics observed in the literature. Then, we used these regular expressions to filter our T-DAGs. \\
After this step, we discarded the graphs listed below as generated by non-standard locking scripts observed in~\cite{bistarelli_2019}:
\begin{itemize}
    \item 2 T-DAGs generated by \textit{P2PKH NOP} transactions, that have been probably used to test the OP\_NOP operator;
    \item 2 T-DAGs generated by \textit{OP\_MIN OP\_EQUAL} transactions, that anyone can easily unlock without any private key;
    \item 5 T-DAGs generated by \textit{Pay to Hash} transactions, considered as ``contest'' in the network to find the correct value of the hash in the transactions;
    \item 1 T-DAGs generated by \textit{OP\_IF} transactions, that could be used to make a P2SH that can be unlocked by revealing only the redeem script;
\end{itemize}
Moreover, we discarded also the following graphs generated by trivial locking scripts never observed before:
\begin{itemize}
    \item 2 T-DAGs generated with the \textit{OP\_CHECKMULTISIG} operator, that anyone can easily unlock without any private key;
\end{itemize}
The remaining 30,221 T-DAGs could be generated either by an unknown transaction never observed before or by a non-standard transaction known in the literature but with non-trivial semantics. Such graphs could be linked to malicious behavior and deserve a careful investigation in future work.
}

\section{Conclusion and Future Work} \label{sec:conclusion}
We have shown that the current assumption that each transaction output has an address attached to it, is false. We have identified transactions that violate the cited assumption, labelling them \emph{unknown} transactions. These unknown transactions imply, among other, that current clustering techniques are incomplete. Starting from  the above  observation, we proposed a theoretical model rooted on sound graph theory to detect, study, and classify patterns of unknown transactions. 
Exploring the Bitcoin network via our new tool we unveiled non-trivial classes of transaction patterns never considered before. 
\ndssColorTex{We were able to identify over 30,000 unknown TX T-DAGs. Each of them represents a money flow that is invisible to standard parsing techniques. Some of these patterns show a high level of sophistication, with a complex topology potentially associated with automated payment services.}
Our novel approach to the Bitcoin graph opens up a brand new vein of research.
For instance, the semantic associated to the discovered patterns is still to be explored. 
Furthermore, by extending the theoretical model to every transaction (using for example flow theory to remove cycles in the user network or, in general, focusing on transaction patterns that are represented by T-DAGs), it could be possible to study and cluster other types of non-standard behaviors within the \texttt{Bitcoin} environment. Moreover, our solution could be adapted to detect anomalous behaviour in the flourishing field of blockchain-based applications.
In conclusion, we believe that the contribution provided in this work, from both a theoretical and practical point of view, other than being interesting on their own---providing for the first time a complete view of the bitcoin blockchain---, 
also pave the way for further research and applications in the Bitcoin domain and its spin-off technologies.


\section*{Acknowledgment}
This publication was partially supported by the award NPRP
11S-0109-180242 from the Qatar National Research Fund
(QNRF), a member of The Qatar Foundation. The information
and views set out in this publication are those of the authors
and do not necessarily reflect the official opinion of the QNRF.

\balance
\bibliography{references}{}
\bibliographystyle{IEEEtranS.bst}



%



\end{document}